\documentclass[a4paper,amsmath,amssymb,aps,prl,twocolumn,preprintnumbers,nofootinbib]{revtex4-2}
\pdfoutput=1 

\usepackage{physics}
\usepackage{braket}
\usepackage{bm}
\usepackage{amsmath,amssymb,amsfonts}
\usepackage{amsthm}
\usepackage{mathtools}
\usepackage{graphicx}
\usepackage{xcolor}
\usepackage[colorlinks=true,
            allcolors=blue]{hyperref}
\usepackage[capitalize,nameinlink]{cleveref}
\usepackage{comment}
\usepackage{dsfont}
\usepackage{enumitem}
\usepackage{bbold}

\newtheorem{theorem}{Theorem}

\theoremstyle{definition}

\newcommand{\cO}{\mathcal O}
\newcommand{\cH}{\mathcal H}
\newcommand{\cA}{\mathcal A}
\newcommand{\cB}{\mathcal B}
\newcommand{\cC}{\mathcal C}
\newcommand{\cE}{\mathcal E}
\newcommand{\cF}{\mathcal F}
\newcommand{\cG}{\mathcal G}

\newcommand{\cR}{\mathcal R}
\newcommand{\cS}{\mathcal S}
\newcommand{\cV}{\mathcal V}
\newcommand{\cI}{\mathcal I}
\newcommand{\IR}{\mathrm{IR}}
\newcommand{\UV}{\mathrm{UV}}

\newcommand{\id}{\mathbb{1}}

\newif\ifclean
\cleanfalse

\newcommand{\tra}[1]{\operatorname{Tr}\!\left[#1\right]}

\begin{document}

\preprint{RUP-26-23}

\title{Informational and algebraic renormalization group}

\author{Takato Mori}\email{takato.mori@yukawa.kyoto-u.ac.jp}
\affiliation{Department of Physics, Rikkyo University, 3-34-1 Nishi-Ikebukuro, Toshima-ku, Tokyo 171-8501, Japan}
\affiliation{RIKEN Quantum, Transformative Research Innovation Platform of RIKEN platforms (TRIP) Headquarters, RIKEN, Wako 351-0198, Japan}

\author{Teruaki Nagasawa}\email{teruaki-nagasawa@se.kanazawa-u.ac.jp}
\affiliation{Institute of Science and Engineering, Kanazawa University, Kanazawa, Ishikawa, 920-1192, Japan}

\date{\today}

\begin{abstract}
    Renormalization group (RG) is a core concept from statistical mechanics to quantum field theory, yet its schemes differ widely across literature. We formulate Wilsonian RG for generic, including finite-dimensional, quantum systems as a quantum channel combining coarse graining with isometric rescaling, whose Kraus representation yields pure conditional trajectories for pure inputs and recovers mixed-state flows upon averaging. We then develop an algebraic extension without assuming tensor-product factorization: the retained low-energy algebra is constructed from the one-particle spectrum, and coarse graining is a conditional expectation onto it. Identifying the channel's Ces\`aro mean with a resource-destroying map onto fixed states establishes a resource theory of RG and yields a relative-entropy monotone without variational optimization. For thermal states in the zero-temperature limit, the resulting Bogoliubov--Kubo--Mori metric per unit spacetime volume is proportional to the Zamolodchikov metric, directly relating the resource density to $c-c_{\rm IR}$ near a two-dimensional IR fixed point. We also introduce the relative entropy of coarse graining to measure discarded UV information, prove its monotonicity along successive coarse grainings, and evaluate it holographically in global $\mathrm{AdS}_3$.
\end{abstract}

\maketitle

\noindent{\em Introduction.}---
The renormalization group (RG) explains how effective descriptions and universal behavior emerge from microscopic physics~\cite{Wilson:1973jj,Polchinski:1983gv}. Real-space schemes act on lattice states or partition functions~\cite{Kadanoff:1966wm,White:1992zz,Vidal:2006sxo}, but their elementary blocking steps are generally discrete; continuous versions~\cite{Haegeman:2011uy,Nozaki:2012zj} are formulated particularly for field theories in the thermodynamic limit. Wilsonian RG eliminates momentum modes and rescales, generating a continuous flow of effective actions or couplings~\cite{Polchinski:1983gv,Wetterich:1992yh} rather than of states. However, it presupposes a continuum description for the rescaling and an identification of the physical scale with the mode momentum. These assumptions become less natural in finite systems, on lattices with nonmonotonic dispersion~\cite{Nielsen:1980rz,Nielsen:1981hk}, or in regimes without quasiparticles. This raises a basic question: can Wilsonian RG and its continuous flow be formulated on the states and observables of generic quantum systems, including finite-dimensional ones?

Answering this question requires specifying what information RG discards. This loss makes RG irreversible: states distinguishable in the ultraviolet (UV) can become indistinguishable in the infrared (IR). The $c$-, $F$-, and $a$-theorems quantify irreversibility under field-theoretic assumptions~\cite{Zamolodchikov:1986gt,Casini:2004bw,Jafferis:2011zi,Casini:2012ei,Komargodski:2011vj,Giombi:2014xxa}. A common framework should make this information loss quantifiable beyond that setting.

Quantum-information approaches have connected RG to distinguishability~\cite{Beny:2012qh,Beny:2014sna} and error correction~\cite{Furuya:2020tzv,Furuya:2021lgx}. Martins Costa \emph{et al.}\ formulated Wilsonian RG as a quantum channel with a momentum-space partial trace and a unitary rescaling~\cite{MartinsCosta:2022bjr}. Goldman, Lashkari, and Leigh derived a Lindbladian for the exact RG of density operators, making any distinguishability measure an RG monotone~\cite{Goldman:2024cvx}. Lashkari built RG monotones from scaling and recovery maps~\cite{Lashkari:2017rcl}. These results motivate a systematic principle for selecting a measure of information loss from RG itself. Resource theory offers a natural candidate: it derives monotones from specified free states and operations without relying on the details of the system~\cite{Chitambar-Gour2019resource-theories}. A resource theory of RG, recently posed as an open problem~\cite{Goto:2026mbj}, would provide this principle by identifying the free states and operations within RG.

In this Letter, we give a unified formulation of Wilsonian RG on the states and observables of generic quantum systems, including finite-dimensional ones. Our informational RG (IRG) combines coarse graining with isometric rescaling into a quantum channel (Fig.~\ref{fig:rg-overview}(a)), accommodating cutoff-regulated Hilbert spaces of unequal dimensions. Its Kraus representation provides a trajectory description of mixed-state RG~\cite{Sang:2023rsp}, allowing comparisons with pure-state approaches~\cite{White:1992zz,Vidal:2006sxo}. Algebraic RG (ARG) replaces an assumed tensor factorization by a retained low-energy algebra constructed from the one-particle spectrum~\cite{Mori:wip}. A conditional expectation onto this algebra, when it exists, coarse grains observables, and its dual channel coarse grains states. Neither construction requires the thermodynamic limit, so both apply to finite-dimensional discrete systems.

\begin{figure}[t]
    \centering
    \includegraphics[width=0.8\columnwidth]{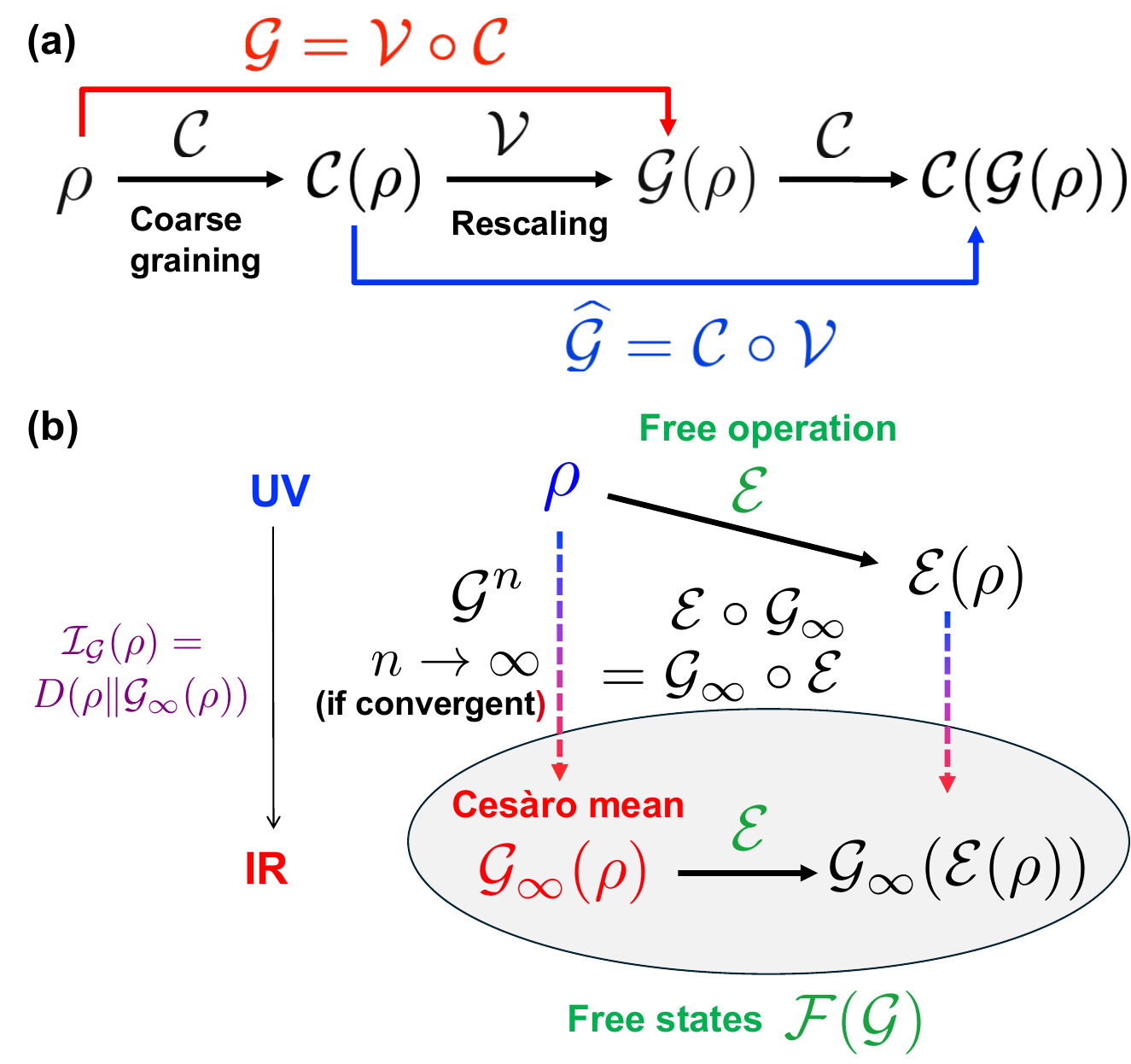}
    \caption{RG as a quantum channel and its resource theory.
    (a) One RG step $\cG=\cV\circ\cC$ coarse grains a state ($\cC$) and rescales it back to the original Hilbert space by an isometry ($\cV$). Iteration proceeds through the reduced map $\widehat\cG=\cC\circ\cV$ on the coarse-grained space.
    (b) Resource theory of RG. The Ces\`aro mean $\cG_\infty$ maps a UV state $\rho$ to a fixed point in the set $\cF(\cG)$ of free states (IR), even when the iterates $\cG^n(\rho)$ do not converge. Free operations $\cE$ commute with $\cG_\infty$ and cannot increase the resource $\cI_\cG(\rho)=D(\rho\|\cG_\infty(\rho))$.}
    \label{fig:rg-overview}
\end{figure}

The RG channel alone fixes a resource theory (Fig.~\ref{fig:rg-overview}(b)), whose free states are RG fixed points. The standard relative-entropy resource measure minimizes the relative entropy of a state with respect to all free states. In finite dimensions, the channel's long-time (Ces\`aro) average projects every state onto a fixed point, even for limit cycles~\cite{Wilson:1970ag,Glazek:2002hq,Bulycheva:2014twa}. Such projections appeared in real-space RG~\cite{Furuya:2020tzv}; here we identify their role as resource-destroying maps~\cite{liu2017resource}. Under the required faithfulness and support conditions, we show that the minimizing free state is uniquely given by the state's own Ces\`aro mean. RG thus selects the optimal reference without variational optimization, and the relative entropy to this reference is the resource measure that never increases along the flow.

The resource-theoretic monotone can be further related to physical RG monotones.
For thermal states in the zero-temperature and thermodynamic limits, the Bogoliubov--Kubo--Mori (BKM) information metric per unit spacetime volume is proportional to the Zamolodchikov metric near a two-dimensional IR fixed point, relating the resource density directly to $c-c_\IR$ along a stable eigen-direction. Complementarily, the relative entropy of coarse graining (RECG), measuring information discarded from a fixed UV state, grows under successive coarse grainings forming a semigroup and reduces to the UV--IR mutual information upon reference optimization, which we also evaluate holographically in global $\mathrm{AdS}_3$.
Wilsonian RG thus extends to finite-dimensional systems, where its irreversibility is the loss of a resource defined by RG itself.

\medskip
{\em Wilsonian RG as a quantum channel.}---
The Wilsonian RG~\cite{Wilson:1973jj,Polchinski:1983gv} consists of two steps: coarse-graining, which integrates out the high-momentum modes in the path integral, and rescaling, which restores the cutoff to its original value. After a field redefinition, one obtains an effective action with the same cutoff but different couplings; the RG flow is the trajectory of the couplings under iteration.

In~\cite{MartinsCosta:2022bjr}, its quantum-channel formulation was considered.\footnote{A complementary approach is done by~\cite{Goldman:2024cvx}, where the system RG evolution is described by a Lindbladian.} For free fields, each momentum mode is an independent oscillator, so one assumes a tensor factorization $\cH=\bigotimes_k\cH_k$. Splitting the modes at a scale $\Lambda'$ defines $\cH=\cH_<\otimes\cH_>$ with $\cH_<=\bigotimes_{k\le\Lambda'}\cH_k$ and $\cH_>=\bigotimes_{k>\Lambda'}\cH_k$; integrating out the high-momentum modes is the partial trace $\Tr_>$, and the rescaling was treated as a unitary channel, being a mere change of variables in the path integral.

However, unitary rescaling requires continuous momenta without a hard IR cutoff; in finite-volume lattice systems, $\cH$ and $\cH_<$ have unequal dimensions. We address this by defining coarse graining on the same Hilbert space by supplementing a reference state and implementing rescaling as an isometric channel, which also connects naturally to conditional expectations.

\medskip
{\em Coarse-graining channel.}---
We supplement a fixed high-momentum reference state $\sigma_>\in\cS(\cH_>)$ to define the coarse-graining channel $\cR_\sigma$ as
\begin{equation}
    \cR_\sigma(\rho)=(\Tr_{>}\rho)\otimes\sigma_>.
    \label{eq:trace-prepare}
\end{equation}
The partial trace retains low-momentum observables while $\sigma_>$ represents the inferred high-momentum state.

This inevitably maps pure states to mixed states if UV--IR entanglement is present. While mixed-state RG flows have been considered in~\cite{Sang:2023rsp}, there also exists a pure-state RG flow such as the density matrix renormalization group (DMRG)~\cite{White:1992zz,White:1993zza} and multi-scale entanglement renormalization ansatz (MERA)~\cite{Vidal:2006sxo}. To understand the relation between these two types of RG flows in the Wilsonian perspective, we rewrite the coarse-graining channel in the Kraus representation, indexed by unitaries $U_>$ on $\cH_>$:
\begin{align}
    \cR_\sigma(\rho)&=\int dU_>\,K_U^{(\sigma)}\rho\,K_U^{(\sigma)\dagger}\\
    K_U^{(\sigma)} &=\sqrt{d_>}\,(\id_<\otimes\sqrt{\sigma_>})(\id_<\otimes U_>),
    \label{eq:mixed-haar}
\end{align}
where $d_>=\dim\cH_>$ and $dU_>$ is normalized Haar measure on $\cH_>$.\footnote{A unitary 1-design also suffices~\cite{Roberts:2016hpo}.}
Intuitively, the Wilsonian coarse-graining first scrambles the high-momentum modes to erase their information and then prepares the fixed reference state $\sigma_>$.
The Kraus representation bridges mixed-state and pure-state RGs: for a pure input, each nonzero Kraus outcome gives a pure conditional state, while its average reproduces $\cR_\sigma(\rho)$. 

\medskip
{\em Rescaling channel.}---
Iterating the coarse-graining with successively lowered cutoffs would replace the entire state by the reference state. To generate a nontrivial flow, the lowered cutoff $\Lambda'$ must be brought back to $\Lambda$. We represent this purity-preserving rescaling by an isometric channel $\cV(\rho)=V\rho V^\dagger$ with $V^\dagger V=\id$, which guarantees that the low-momentum observables $\cB(\cH_<)$ are preserved.

In infinite dimensions the rescaling $p\to rp$, $r=\Lambda/\Lambda'>1$, acts locally in $p$; we similarly construct a global rescaling from local isometries. Let $\cH_<=\bigotimes_{i=0}^{N-1}\cH_i'$ carry $N$ momentum modes, to be rescaled to $\widetilde\cH\cong\cH=\bigotimes_{a=0}^{M-1}\widetilde\cH_a$ with $M>N$ modes. Partition the target sites as $\{0,\dots,M-1\}=\bigsqcup_iC_i$ and choose orthonormal states $\{\ket{\Phi^{(i)}_s}\}_s\subset\widetilde\cH_{C_i}:=\bigotimes_{a\in C_i}\widetilde\cH_a$; then the local maps $V_i=\sum_s\ket{\Phi^{(i)}_s}\bra{s}_i$ are isometries and $V=\bigotimes_iV_i$ is a global isometry. For instance, a piecewise-constant interpolation in the momentum basis corresponds to $V_i=\sum_s\ketbra{ss}{s}$. The field redefinition is a local basis change absorbed in the choice of $\ket{\Phi^{(i)}_s}$.

For the vacuum of a free scalar theory in Minkowski spacetime, the rescaling becomes unitary and the momentum rescaling together with the field redefinition is given by a two-mode squeezing unitary 
$S_{m'\leftarrow m}=\bigotimes_{\bm p}\exp[\frac14\log\frac{\omega_{m'}(\bm p)}{\omega_m(\bm p)}\left(a_m(\bm p)a_m(-\bm p)-a^\dagger_m(\bm p)a^\dagger_m(-\bm p)\right)]$
implementing the mass shift $m^2\to {m^\prime}^2=r^2 m^2$.
Here $\omega_m(\bm p)=\sqrt{\bm p^2+m^2}$ and the momentum product is understood with a chosen regulator. In the thermodynamic limit the rescaling merely compresses the sparse coarse-grained lattice, while the field redefinition restoring the canonical kinetic term induces the mass shift via on-site unitaries.

\medskip
{\em Low-energy algebra.}---
The channel formulation above applies to discrete systems, but it rests on two assumptions: a tensor-factorized Hilbert space, which becomes subtle in the continuum or thermodynamic limit and away from free fields, and the identification of mode momentum with the energy scale. The latter holds for relativistic continuum QFTs in Minkowski spacetime but can fail otherwise; e.g., on a lattice with spacing $\epsilon$ the momentum lives in the Brillouin zone with $2\pi/\epsilon$ periodicity, so it need not scale with the mode energy even under translational invariance and locality~\cite{Nielsen:1980rz,Nielsen:1981hk}.

To overcome these issues, we treat the low-energy modes algebraically rather than as a tensor factor. The key idea is to identify the low-energy algebra from the one-particle spectrum: a collective excitation carrying large total energy need not lie in the high-energy algebra---classical electromagnetism describes such collective phenomena---and what matters is the mode spectrum of the one-particle Hamiltonian.
We propose a six-step procedure: (1) choose single-particle generating operators representing elementary excitations relevant to RG; (2) take their linear span as the one-particle Hilbert space; (3) project the total Hamiltonian onto it to define the one-particle Hamiltonian; (4) \textbf{coarse-graining}: truncate the one-particle spectrum and form the one-particle coarse-grained algebra from its action on a reference state; (5) \textbf{rescaling}: apply an isometry realizing piecewise interpolation through a superposed basis; and (6) take the double commutant (adding the identity and adjoints) to define the low-energy multiparticle algebra $\cA_L$, with the high-energy algebra given by its relative commutant.
One can confirm that this reproduces the low-energy algebra for free fields, $O\otimes \id\in \cB(\cH_<\otimes\cH_>)$. For details, refer to the companion paper~\cite{Mori:wip}.

\medskip
{\em Coarse graining as a conditional expectation.}---
Once the retained low-energy algebra is specified, the coarse-graining part of RG can be formulated as a conditional expectation~\cite{takesaki1972conditional,Furuya:2020tzv}.

Let $\cA\subseteq\cB(\cH)$ be a von Neumann algebra of observables and $\cA_L\subseteq\cA$ a retained von Neumann subalgebra with the same identity. We assume that a normal conditional expectation $E:\cA\to\cA_L$ exists. Viewed as a map on $\cA$ via the inclusion $\cA_L\subseteq\cA$, it is a normal, unital, completely positive projection that fixes $\cA_L$ pointwise and satisfies the bimodule property $E(a_Lab_L)=a_LE(a)b_L$ for $a\in\cA$ and $a_L,b_L\in\cA_L$. The retained algebra alone need not determine $E$ uniquely.

For $\cA=\cB(\cH_<\otimes\cH_>)$ and $\cA_L=\cB(\cH_<)\otimes\mathbb C\id_>$, a reference state $\sigma_>$ defines
$E_{\sigma_>}(a)=\Tr_>[(\id_<\otimes\sigma_>)a]\otimes\id_>$, which preserves low-energy observables and averages the discarded sector in $\sigma_>$. While $E$ acts on observables, its dual channel acts on states as $E_*(\omega)=\omega\circ E$, recovering $E_{\sigma_>*}(\rho)=\cR_\sigma(\rho)=(\Tr_>\rho)\otimes\sigma_>$ in Eq.~\eqref{eq:trace-prepare}.
In the following, we identify a linear functional $\omega_\rho$ with the corresponding density matrix $\rho$ unless noted.

While a coarse-graining step is a conditional expectation, a single RG step is not, since idempotency fails; for example, the repetition isometry yields a dephasing channel preserving only the diagonal algebra. Nevertheless, a certain long-time average of RG will become a conditional expectation, eventually viewed as a resource-destroying map~\cite{liu2017resource}.

\medskip
{\em Resource theory of RG.}---
A defining feature of RG is its irreversibility: UV information is lost toward the IR, as captured by monotones such as Zamolodchikov's $c$-function~\cite{Zamolodchikov:1986gt} and its entropic versions~\cite{Casini:2004bw,Casini:2006es}, which are, however, constructed case by case. Resource theories provide a unifying principle: once free states and free operations are specified, contractive distances from the free set furnish monotones~\cite{Chitambar-Gour2019resource-theories}. The formulations above allow us to apply this idea directly to RG.

Let $\cC:\cS(\cH)\rightarrow\cS(\cH_<)$ be the coarse-graining map (e.g., $\cC=\Tr_>$ in a tensor-factorized description) and let
$\cV(\tau)=V\tau V^\dagger:\cS(\cH_<)\rightarrow\cS(\cH)$ ($V^\dagger V=\id_<$, $P=VV^\dagger$) be the rescaling channel. Hereafter a dagger denotes the Heisenberg dual. One full RG step $\cG:=\cV\circ\cC$ and the reduced channel $\widehat{\cG}:=\cC\circ\cV$ satisfy
\begin{equation}
    \cG^n
    =
    \cV\circ
    \widehat{\cG}^{n-1}
    \circ\cC,
    \quad n\geq1.
    \label{eq:RG-power-factorization}
\end{equation}
Thus, all nontrivial iteration takes place in the coarse-grained space $\cH_<$.

The free states of the full RG map are its fixed-point sets, namely,
\begin{equation}
    \cF(\cG)
    :=
    \{\rho\mid\cG(\rho)=\rho\}=\cV\!\left(\cF(\widehat\cG)\right).
    \label{eq:rg-free-states}
\end{equation}
The resource-destroying map~\cite{liu2017resource} is the Ces\`aro mean, which by Eq.~\eqref{eq:RG-power-factorization} factorizes via the reduced mean $\widehat\cG_\infty$:
\begin{equation}
    \cG_\infty
    :=
    \lim_{N\to\infty}
    \frac1N\sum_{n=1}^N\cG^{\,n}
    =
    \cV\circ\widehat\cG_\infty\circ\cC.
    \label{eq:cesaro-factorization}
\end{equation}
These averages exist in finite dimensions even when the iterates oscillate, as in limit cycles~\cite{Wilson:1970ag,Glazek:2002hq,Bulycheva:2014twa}. They are channel projections onto the fixed states~\cite{wolf2012qchannels,petz2008quantum},
$\cG_\infty^2=\cG\cG_\infty=\cG_\infty\cG=\cG_\infty$,
so averaging removes the resource while preserving every free state.

Suppose that $\widehat{\cG}$ admits a full-rank fixed state $\widehat\rho_*$.\footnote{Every output of $\cG$ is supported on $P$. Thus, for $P\neq\id$, no fixed state is full rank on $\cH$. The reduced assumption instead makes $V\widehat\rho_*V^\dagger$ faithful on $P\cB(\cH)P$.} Then the fixed-point observables form an algebra~\cite{wolf2012qchannels}
\begin{equation}
    \widehat{\cA}_{\cG}
    :=
    \left\{
    a\in\cB(\cH_<)
    \,\middle|\,
    \widehat{\cG}^{\,\dagger}(a)=a
    \right\}.
    \label{eq:reduced-fixed-algebra}
\end{equation}
Here $\widehat\cG_\infty^\dagger$ becomes a conditional expectation onto $\widehat\cA_{\cG}$, which we denote as $\widehat\Pi_\cG$~\cite{takesaki1972conditional}. 
The Heisenberg-picture Ces\`aro mean of the full RG map is
\begin{equation}
    \Pi_\cG(a)
    =(\cC^\dagger\circ\widehat\Pi_\cG\circ\cV^\dagger)(a)
    =\widehat\Pi_\cG(V^\dagger aV)\otimes\id_>,
    \label{eq:full-cesaro-CE}
\end{equation}
where the last equality holds for $\cC=\Tr_>$, defining a conditional expectation onto the fixed-point algebra $\cA_\cG=\widehat\cA_\cG\otimes\mathbb C\id_>$.
Invariant-algebra projections have also appeared in real-space RG~\cite{Furuya:2020tzv}, microscopicity~\cite{nagasawa2025macroscopicity}, and related resource theories~\cite{yoeli2026instability}.
In this partial-trace setting with a faithful reduced fixed state, the Heisenberg-picture Ces\`aro mean is a conditional expectation, although the dual of a single RG step generally is not.

The free operations are channels $\cE$ commuting with $\cG_\infty$. This includes every channel commuting with $\cG$, and hence $\cG$ itself.
Let $D_{\cA}$ be the Araki relative entropy~\cite{araki1976relative}, which reduces in finite dimensions to the Umegaki relative entropy $D(\rho\|\sigma)=\tra{\rho(\log\rho-\log\sigma)}$~\cite{umegaki1962conditional}.
We define the \emph{relative entropy of RG}
\begin{equation}
    \cI_\cG(\rho)
    :=
    D_\cA\!\left(
    \rho\,\middle\|\,\cG_\infty(\rho)
    \right).
    \label{eq:rg-monotone}
\end{equation}
It compares a state with the fixed-point state of its own RG trajectory.

\begin{theorem}[RG resource monotone]
\label{th:rg-resource-monotone}
For the finite-dimensional construction above:

(i) $\cI_\cG(\rho)\geq0$, with equality exactly on $\cF(\cG)$;

(ii) $\cI_\cG(\cE(\rho))\leq\cI_\cG(\rho)$ for every free operation $\cE$;

(iii) if $\widehat\cG$ has a faithful fixed state, then
\begin{equation}
    R(\rho):=\min_{\eta\in\cF(\cG)}D_\cA(\rho\|\eta)
    =\cI_\cG(\rho).
    \label{eq:rg-min-distance}
\end{equation}
For $\rho=P\rho P$, the minimum is finite and uniquely attained at $\cG_\infty(\rho)$. Otherwise all comparisons are infinite.
\end{theorem}
\noindent
Parts (i)--(ii) need no faithful fixed state. Part (iii) uses only the reduced expectation, so it holds for general $\cC$. In particular,
\begin{equation}
    \cI_\cG(V\tau V^\dagger)
    =D\!\left(\tau\,\middle\|\,\widehat\cG_\infty(\tau)\right).
    \label{eq:rg-supported-resource}
\end{equation}
Every input reaches this support after one RG step. Under (iii), the relative-entropy resource measure~\cite{Chitambar-Gour2019resource-theories} is therefore obtained by averaging, without an optimization over free states. Proofs are given in the Supplemental Material.

\medskip
{\em Relation to the $c$-theorem.}---
We now relate the zero-temperature limit of the resource monotone for Gibbs states to the $c$-function, an RG monotone discussed in physics.
While the regularization prescription can be arbitrary, we choose it so that $\sqrt{\tau^2+\abs{\bm{x}}^2}> \epsilon$ is the physical region.
The detailed derivation of the following equations is presented in the Supplemental Material.

Start with faithful thermal states $\rho(g)=e^{-\beta H(g)}/Z(g)$ on the low-energy Hilbert space $\cH_<$, with
\begin{equation}
    H(g)=H_0+\sum_i g^i\int_\Sigma d^{d-1}\bm x\,\cO_i(0,\bm x).
\end{equation}
Here $\{g^i\}_i$ are the coupling constants associated with operators $\{\cO_i\}_i$, and $\Sigma$ is a spatial slice. At each finite inverse temperature $\beta$, spatial volume, and regularization prescription, the relative entropy expands as
\begin{equation}
    D\big(\rho(g+dg)\big\|\rho(g)\big)=\tfrac12\chi_{ij}(g)\,dg^idg^j+O(dg^3),
    \label{eq:BKM-expansion}
\end{equation}
where $\chi_{ij}$ is the BKM information metric~\cite{petz2008quantum}. With $h_i=\partial_iH$, $\delta h_i=h_i-\expval{h_i}_\rho$, and $\Omega_\rho(A)=\int_0^1dt\,\rho^tA\rho^{1-t}$, one has $\partial_i\rho=-\beta\Omega_\rho(\delta h_i)$ and
\begin{equation}
    \chi_{ij}=\beta\int_0^\beta\!d\tau\!\int_\Sigma\!d^{d-1}\bm x\,d^{d-1}\bm y\,
    \expval{\cO_i(\tau,\bm x)\cO_j(0,\bm y)}_{\beta,c}.
    \label{eq:BKM-Euclidean-correlator}
\end{equation}
The subscripts $\beta,c$ denote the finite-temperature, connected correlator. 

In a two-dimensional unitary field theory, vacuum correlators of the operators conjugate to the couplings define the Zamolodchikov metric~\cite{Zamolodchikov:1986gt,Friedan:2009ik}, and
\begin{equation}
    \frac{dc}{ds}=-\kappa g^{\rm Zam}_{ij}\beta^i\beta^j,
    \qquad \beta^i=\frac{dg^i}{ds}.
    \label{eq:c-function-slope}
\end{equation}
Here $s=-\log\Lambda$ increases toward the IR, $\kappa>0$ fixes the normalization, and the $c$-function equals the central charge at fixed points.

Assume that a fixed $V$ and $\widehat\cG$ realize a physical RG trajectory $g^i(u)=g_*^i+uv^i+O(u^2)$ with an IR fixed-point state such that $\widehat\cG_\infty(\rho(g(u)))=\rho_*:=\rho(g_*)>0$. We further assume that $v$ is a stable direction at $g_*$, $du/ds=-\omega u+O(u^2)$ with stable exponent $\omega>0$, terminating at the fixed point $g_*$.
Eq.~\eqref{eq:rg-supported-resource} gives
\begin{equation}
    \cI_\cG(V\rho(g(u))V^\dagger)
    =D(\rho(g(u))\|\rho_*)
    =\tfrac12\chi_{vv*}u^2+O(u^3),
    \label{eq:I-Zam}
\end{equation}
where $\chi_{vv*}=v^i\chi_{ij}(g_*)v^j$.

Suppose that $v$ corresponds to a scalar operator $\cO_v=\sum_i v^i \cO_i$ whose scaling dimension is $\Delta_v=2+\omega>2$ at the IR conformal fixed point, thus $\cO_v$ is irrelevant there. 
For dimensionless $u$ at a fixed reference scale $\mu$, the (dimensionless) Zamolodchikov metric is defined through ($z=x+i\tau$)
\begin{align}
    & \mu^{2\omega}\expval{\cO_v(z)\cO_v(0)}_{0,c}
    =\frac{g^{\rm Zam}_{vv*}}{|z|^{2\Delta_v}},\\
    & g^{\rm Zam}_{vv*}
    =\left.\mu^{-4}\expval{\cO_v(z)\cO_v(0)}_{0,c}
    \right|_{|z|=\mu^{-1}}.
    \label{eq:Zam-fixed-point-normalization}
\end{align}
An overall metric convention is absorbed into $\kappa$.

Assuming that the spacetime integral and the zero-temperature limit commute, the BKM metric per unit spacetime volume becomes, after some elementary calculus,
\begin{equation}
\begin{aligned}
    \hat{\chi}_{vv*}
    &:=\lim_{L\to\infty}\lim_{\beta\to\infty}
        \frac{\chi_{vv*}}{\beta L}
    =\int_{|z|>\epsilon}d^2z\,
        \expval{\cO_v(z)\cO_v(0)}_{0,c}\\
    &=\frac{\pi\mu^2}{(\Delta_v-1)(\mu\epsilon)^{2(\Delta_v-1)}}
        g^{\rm Zam}_{vv*} =: \mu^2 C_v(\mu \epsilon)\, g^{\rm Zam}_{vv*}.
\end{aligned}
    \label{eq:BKM-Zam}
\end{equation}

From Eq.~\eqref{eq:I-Zam}, this yields the following relation between the resource monotone and the Zamolodchikov metric per unit volume:
\begin{equation}
    \lim_{L,\beta\to\infty}\frac{\cI_\cG(V\rho(g(u))V^\dagger)}{\beta L}
    = \frac{1}{2}\mu^2 C_v(\mu \epsilon)\,
        g^{\rm Zam}_{vv*} u^2 + O(u^3).
    \label{eq:IG-Zam-density}
\end{equation}
Integrating Eq.~\eqref{eq:c-function-slope} gives
\begin{equation}
    c(u)-c_{\rm IR}
    =\frac{\kappa\omega}{2}g^{\rm Zam}_{vv*}u^2+O(u^3),
    \label{eq:c-near-IR-direction}
\end{equation}
and hence
\begin{equation}
    \lim_{L,\beta\to\infty}\frac{\cI_\cG(V\rho(g(u))V^\dagger)}{\beta L}
    = \frac{\mu^2 C_v(\mu \epsilon)}{\kappa \omega}
    \bigl(c(u)-c_{\rm IR}\bigr)+O(u^3).
    \label{eq:IG-c-near-IR}
\end{equation}
This relation directly relates the resource monotone to the $c$-function, a physical monotone under RG.

\medskip
{\em Monotone from coarse graining.}---
We now fix the UV state $\rho$ and vary the retained algebra $\cA_\Lambda\subseteq\cA$. Choose a conditional expectation $E_\Lambda:\cA\to\cA_\Lambda$ at each cutoff and set $\rho_\Lambda=E_{\Lambda*}(\rho)$. The \emph{relative entropy of coarse-graining} (RECG) is
\begin{equation}
    \cI_{E_\Lambda}(\rho)=D_\cA(\rho\|\rho_\Lambda).
    \label{eq:monotone}
\end{equation}
Both states live on the UV algebra and a common isometric embedding leaves the value unchanged. Unlike $\cI_\cG$, which measures a separation from the IR fixed point, RECG measures information discarded from a fixed UV input.

We work in finite dimensions below. For retained and discarded sectors $L$ and $H$, let $\rho_L=\Tr_H\rho$ and $\rho_H=\Tr_L\rho$. The trace-and-prepare channel $E_{\sigma_H*}(\rho)=\rho_L\otimes\sigma_H$ gives
\begin{equation}
    \cI_{E_{\sigma_H}}(\rho)
    =I_\rho(L\!:\!H)+D(\rho_H\|\sigma_H),
    \label{eq:decomposition}
\end{equation}
where $I_\rho(L\!:\!H)=S(\rho_L)+S(\rho_H)-S(\rho)$. The two terms measure discarded UV--IR correlations and the mismatch with the estimated UV prior.

\begin{theorem}[Optimal reference state]
\label{th:optimal-reference}
For a fixed factorization and state $\rho$,
\begin{equation}
    \min_{\sigma_H}\cI_{E_{\sigma_H}}(\rho)=I_\rho(L\!:\!H),
\end{equation}
with the unique minimizer $\sigma_H=\rho_H$.
\end{theorem}
\noindent
The optimal reference depends on the input, so $\rho\mapsto\rho_L\otimes\rho_H$ is generally nonlinear.

We now show monotonicity along a tower of conditional expectations. For cutoffs $\Lambda_2<\Lambda_1$, suppose $\cA_{\Lambda_2}\subseteq\cA_{\Lambda_1}$ and write $E_a=E_{\Lambda_a}$. Inclusion gives $E_1\circ E_2=E_2$. We also require the semi-group condition, $E_2\circ E_1=E_2$.
Taking duals reverses the order: $E_{1*}E_{2*}=E_{2*}$. Thus the coarser state is invariant under the finer coarse-graining.

\begin{theorem}[Coarse-graining monotone]
\label{thm:rg-monotone}
For compatible conditional expectations as above, let $\rho_a=E_{a*}(\rho)$. Then
\begin{equation}
    \cI_{E_2}(\rho)
    =\cI_{E_1}(\rho)+D_\cA(\rho_1\|\rho_2)
    \geq\cI_{E_1}(\rho).
    \label{eq:pythagoras-tower}
\end{equation}
The identity includes infinite values. If $\cI_{E_1}(\rho)<\infty$, equality holds exactly when $\rho_1=\rho_2$.
\end{theorem}
\noindent
This follows from the Pythagorean identity for relative entropy~\cite{petz2008quantum}, applied using $E_{1*}(\rho_2)=\rho_2$. Along a compatible tower starting at $E_{\Lambda_\UV}=\mathrm{id}_\cA$,
$0=\cI_{E_{\Lambda_\UV}}(\rho)\leq\cI_{E_{\Lambda_1}}(\rho)\leq\cI_{E_{\Lambda_2}}(\rho)$.
Independent reference optimization need not preserve compatibility, so the theorem does not establish monotonicity of the optimized mutual information. 
RECG increases toward the IR, whereas the RG resource decreases. Their directions agree with $c_\UV-c$ and $c-c_\IR$, respectively, but the RECG implies no quantitative identification with a monotone $c$-function yet. We discuss a possible holographic interpretation below.

\medskip
{\em Holographic realization.}---In AdS/CFT~\cite{Maldacena:1997re}, the radial direction $r$ encodes the boundary energy scale $\Lambda$. Moving the radial cutoff implements holographic renormalization~\cite{Skenderis:2002wp}. Its connection to Wilsonian RG~\cite{Heemskerk:2010hk,Faulkner:2010jy,Balasubramanian:2012hb,Radicevic:2011py,Sin:2011yh,Kim:2020kip,Grozdanov:2011aa,Akhmedov:2010sw} and multi-scale entanglement renormalization ansatz~\cite{Swingle:2009bg,Nozaki:2012zj} has been studied extensively.

Let the metric of the global AdS$_3$ be
\begin{equation}
    ds^2=-(r^2+1)dt^2 + \frac{dr^2}{r^2+1}+r^2d\theta^2
\end{equation}
with the AdS radius rescaled to unity. In these coordinates, the boundary is located at $r\to\infty$.

The relation between scale and radius motivates us to associate the causal-diamond algebras of the bulk subregions
\begin{equation}
    L:\ r\le\Lambda,\qquad H:\ r>\Lambda,
    \label{eq:radial-split}
\end{equation}
with the low- and high-energy algebras, respectively (Fig.~\ref{fig:hol}). A related interpretation appears in studies of momentum-space entanglement~\cite{Balasubramanian:2011wt}. 

\begin{figure}[t]
    \centering
    \includegraphics[width=0.9\linewidth]{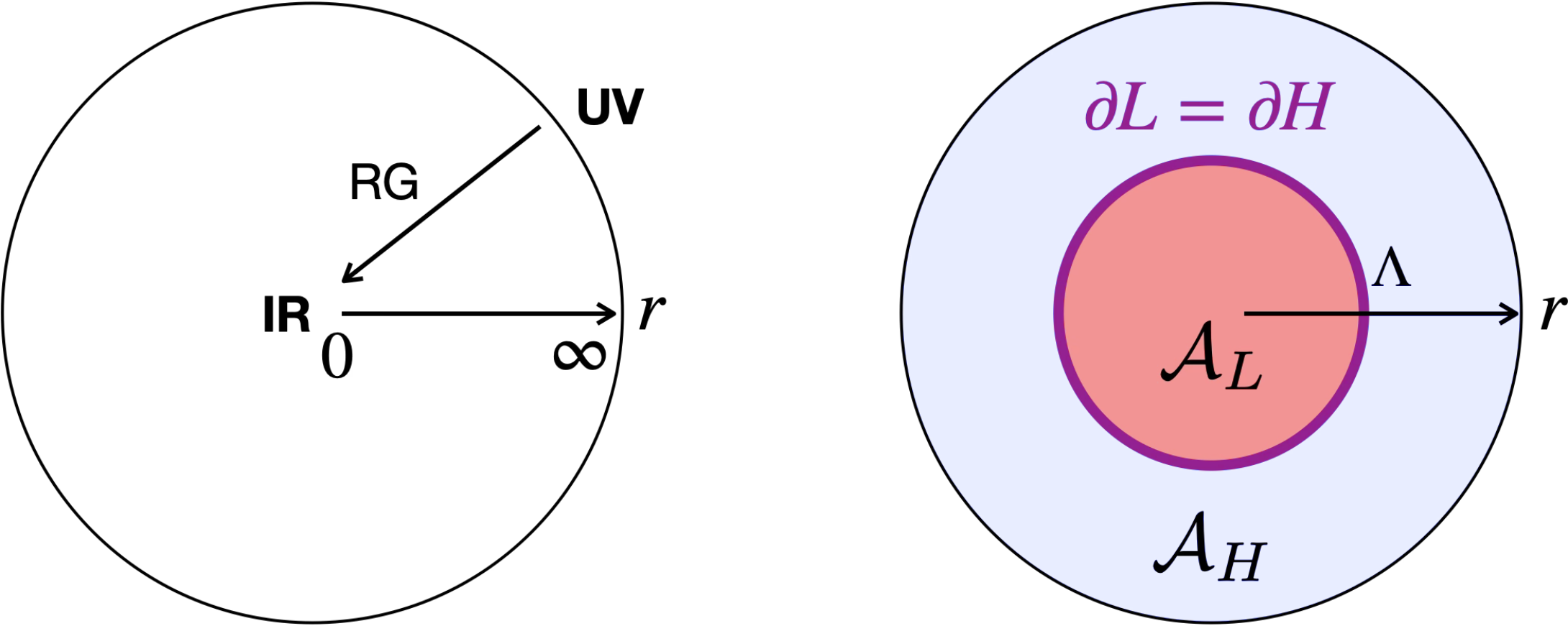}
    \caption{Holographic setup on a time slice of global AdS$_3$. Left: The radial coordinate $r$ sets the RG scale, from the UV at $r=r_\infty$ to the IR at $r=r_{\rm IR}$. 
    Right: $\mathcal{A}_L$ and $\mathcal{A}_H$ are the algebras of $L=\{r\le\Lambda\}$ and its complement $H$, separated at the cutoff $r=\Lambda$.}
    \label{fig:hol}
\end{figure}

Based on this identification, let us compute the RECG holographically. When the product-state construction of Theorem~\ref{th:optimal-reference} applies, the measure is $I_\rho(L\!:\!H)$. Assume that each entropy in the measure is given by the Ryu--Takayanagi formula~\cite{Ryu:2006bv} for the corresponding bulk subregion.\footnote{See also related works known as `hole-ography'~\cite{Balasubramanian:2013lsa,Czech:2014wka,Headrick:2014eia,Espindola:2017jil}.}
Here we consider a pure global vacuum $\rho$ ($S(\rho)=0$; for mixed states, see~\cite{Mori:wip}). Using the Brown--Henneaux relation $c=3/(2G_N)$~\cite{Brown:1986nw},
\begin{equation}
    S(\rho_L)=S(\rho_H)=\frac{\mathrm{Area}(\partial L)}{4G_N} = \frac{\pi c \Lambda }{3},\quad
    I_\rho(L\!:\!H)=\frac{2\pi c}{3}\Lambda.
\end{equation}
Thus the optimized RECG is proportional to the central charge and decreases monotonically as the coarse-graining scale $\Lambda$ is lowered.\footnote{Note that the monotonic decrease of $I_\rho(L:H)$ as $\Lambda$ decreases does not contradict Theorem~\ref{thm:rg-monotone}. The former is after a state-dependent optimization over the reference state, so the theorem does not apply.}

\medskip
\noindent{\em Conclusions.}---
We have developed informational and algebraic RG through coarse graining and isometric rescaling. These formulations connect conditional pure-state trajectories with averaged mixed-state flows and replace tensor-factorized subsystems by retained low-energy algebras. In finite dimensions, the Ces\`aro mean projects onto RG fixed states even for nonconvergent iterates, providing a resource-destroying map. Relative entropy to this average decreases under channels commuting with it and, when the reduced channel has a faithful fixed state, equals the relative entropy to the free set without optimization.
For thermal states in the zero-temperature and thermodynamic limits, the BKM metric per unit spacetime volume is proportional to the Zamolodchikov metric, directly relating the RG resource density to $c-c_{\rm IR}$ along a stable direction near a two-dimensional IR fixed point.
In contrast, the information discarded from a fixed UV state increases along compatible coarse grainings. Its reference-optimized value, the UV--IR mutual information, is proportional to the central charge for bulk radial algebras in global AdS$_3$.

Several questions remain. Relative entropy diverges under support mismatch, as in pure-state flows of free fields; bounded measures based on fidelity or trace distance may avoid this~\cite{Goldman:2024cvx}. Extending the resource theory of RG to type-II and type-III algebras is also essential for continuum QFT and gravity~\cite{MartinsCosta:MomentumSpaceAlgebras}, where normal conditional expectations need not exist. Relating this state-based RG to flows of effective actions via discrete analogues of exact RG equations would extend it from states to theories~\cite{Kuwahara:2022nlm,Goldman:2023gzr}. Other directions include complete conversion criteria under free operations and connections to the holographic $c$-theorem~\cite{Freedman:1999gp}.

\medskip
{\em Use of artificial intelligence tools.}---
The authors used OpenAI ChatGPT (GPT-6.1 Sol, GPT-6 Astra, GPT-6.1 Sol, and GPT-5.6 Sol), Anthropic Claude (Fable 5, Fable 5.1, Opus 4.8, Opus 5.5, and Sonnet 5.5), and Google Gemini (4 Argon and 3.1 Pro) interactively to assist with mathematical exploration and proof strategies, literature surveys, organization and refinement of arguments, manuscript preparation, and figure generation and refinement. The quantum-channel and algebraic formulations of RG and its resource-theoretic treatment originated with the authors. AI-assisted reasoning contributed substantially to developing the connection to the $c$-theorem, including the core proof strategy. The relevant arguments were cross-checked using multiple models, after which the authors independently verified every step and reconstructed and wrote the proofs themselves. All AI-assisted content was reviewed and modified as necessary by the authors, who take full responsibility for the correctness and content of the manuscript.

\medskip
{\em Acknowledgments.}---We thank Matheus H. Martins Costa, Shoji Hashimoto, Kanta Masuki, Javier Moreno, Ryota Nasu, and Taichi Tanaka for helpful comments and discussion.
T.~M. was supported by the RIKEN TRIP initiative and JSPS KAKENHI Grant Numbers 23KJ1154 and 24K17047. T.~N. acknowledges support from JST ERATO Grant Number JPMJER2402. This work was supported by MEXT KAKENHI Grant-in-Aid for Transformative Research Areas A ``Extreme Universe'' Numbers 21H05182, 21H05183, and 21H05187.

\bibliography{library}

\onecolumngrid
\newpage

\setcounter{page}{1}
\renewcommand{\thepage}{Supplemental Material -- \arabic{page}}
\setcounter{equation}{0}
\renewcommand{\theequation}{S.\arabic{equation}}
\renewcommand{\theHequation}{supp.\arabic{equation}}

\begin{center}
    {\large\bf Supplemental Material for\\[2pt] ``Informational and algebraic renormalization group''}\\[6pt]
    Takato Mori and Teruaki Nagasawa
\end{center}

\section*{Proof of Theorems~\ref{th:rg-resource-monotone},~\ref{th:optimal-reference},~\ref{thm:rg-monotone}}

\begin{proof}[Proof of Theorem~\ref{th:rg-resource-monotone}]
    (i) Nonnegativity and the equality condition of the relative entropy give $\mathcal I_\cG(\rho)=0$ iff $\rho=\cG_\infty(\rho)$, which is equivalent to $\rho\in\cF(\cG)$ from the fact $\cG_\infty^2=\cG\cG_\infty=\cG_\infty\cG=\cG_\infty$.
    (ii) Using the data processing inequality,
    \begin{equation}
        \mathcal I_{\cG}(\cE(\rho))
        =D_{\cA}\!\left(\cE(\rho)\,\middle\|\,\cG_{\infty}(\cE(\rho))\right)
        =D_{\cA}\!\left(\cE(\rho)\,\middle\|\,\cE(\cG_{\infty}(\rho))\right)
        \leq D_{\cA}\!\left(\rho\,\middle\|\,\cG_{\infty}(\rho)\right)
        =\mathcal I_{\cG}(\rho).
    \end{equation}
    (iii) Every free state and every output of $\cG_\infty$ is supported on $P=VV^\dagger$. If $\rho\neq P\rho P$, all relative entropies in Eq.~\eqref{eq:rg-min-distance} are infinite. Otherwise write $\rho=V\tau V^\dagger$ and $\eta=V\zeta V^\dagger$, where $\zeta\in\cF(\widehat\cG)$. Equation~\eqref{eq:cesaro-factorization} then gives $\cG_\infty(\rho)=V\widehat\cG_\infty(\tau)V^\dagger$. Every free state $\zeta\in\cF(\widehat\cG)$ obeys $\omega_\zeta=\omega_\zeta\circ\widehat\Pi_\cG$, for the conditional expectation $\widehat\Pi_\cG=\widehat\cG_\infty^\dagger$ onto $\widehat\cA_\cG$, which exists by the faithful reduced fixed-state assumption, so the Pythagorean theorem for the conditional expectation~\cite{petz2008quantum} gives
    \begin{equation}
    \begin{aligned}
        D_{\cA}(\rho\|\eta)&=D(\tau\|\zeta)\\
        &=D\!\left(\tau\,\middle\|\,\widehat\cG_\infty(\tau)\right)+D\!\left(\widehat\cG_\infty(\tau)\,\middle\|\,\zeta\right)\\
        &=\mathcal I_{\cG}(\rho)+D_{\cA}\!\left(\cG_\infty(\rho)\,\middle\|\,\eta\right).
    \end{aligned}
    \end{equation}
    Here we used isometric invariance of relative entropy. Taking the faithful fixed state $\zeta=\widehat\rho_*>0$ makes $D(\tau\|\zeta)$ finite and hence proves finiteness of $\mathcal I_\cG(\rho)$ on $P\cH$.
    Minimizing over $\eta\in\cF(\cG)$, the second term vanishes exactly at $\eta=\cG_\infty(\rho)\in\cF(\cG)$.
\end{proof}

\begin{proof}[Proof of Theorem~\ref{th:optimal-reference}]
    With $\omega_\rho\circ E_{\sigma_H}=\rho_L\otimes\sigma_H$,
    \begin{equation}
        \cI_{E_{\sigma_H}}(\rho)=D(\rho\|\rho_L\otimes\sigma_H)
        =[S(\rho_L)+S(\rho_H)-S(\rho)]+D(\rho_H\|\sigma_H)
        =I_\rho(L\!:\!H)+D(\rho_H\|\sigma_H),
    \end{equation}
    which is Eq.~\eqref{eq:decomposition}. Nonnegativity of $D(\rho_H\|\sigma_H)$, with equality iff $\sigma_H=\rho_H$, proves the theorem.
\end{proof}

\begin{proof}[Proof of Theorem~\ref{thm:rg-monotone}]
    Let $\rho_2:=\omega_\rho\circ E_2$. The semi-group condition $E_2\circ E_1=E_2$ gives
    \begin{equation}
        \rho_2\circ E_1=\omega_\rho\circ(E_2\circ E_1)=\omega_\rho\circ E_2=\rho_2,
    \end{equation}
    so $\rho_2$ is $E_1$-invariant. The Pythagorean theorem of relative entropy with respect to the conditional expectation $E_1$~\cite{petz2008quantum} then gives
    \begin{equation}
        D_\cA(\rho\|\rho_2)=D_\cA(\rho\|\rho_1)+D_\cA(\rho_1\|\rho_2),\qquad \rho_1=\omega_\rho\circ E_1,
    \end{equation}
    which is Eq.~\eqref{eq:pythagoras-tower}.
\end{proof}

\section*{Derivation of the BKM metric for thermal states}
\label{app:BKM}

We first derive Eqs.~\eqref{eq:BKM-expansion} and~\eqref{eq:BKM-Euclidean-correlator} at finite temperature, and then establish the continuum comparison in Eqs.~\eqref{eq:BKM-Zam}--\eqref{eq:IG-c-near-IR}. In parts 1--3, $\rho(g)=e^{-\beta H(g)}/Z(g)>0$, with $Z(g)=\Tr e^{-\beta H(g)}$, acts on a fixed finite-dimensional low-energy Hilbert space $\cH_<$, at fixed $\beta$, spatial volume, and UV regulator. All traces, logarithms, and inverses in these parts are on $\cH_<$, and $\partial_i:=\partial/\partial g^i$. Parts 4--6 concerns the zero-temperature and thermodynamic limits, the two-point correlator with the UV regularization $|z|>\epsilon$, and the derivation of Eqs.~\eqref{eq:c-near-IR-direction} and \eqref{eq:IG-c-near-IR}.

\subsection*{1. Relative entropy at second order}

Let $\sigma(t):=\rho(g+t\,dg)$ with $\sigma(0)=\rho(g)$, and set
$D(\sigma(t)\|\rho)=\Tr[\sigma\log\sigma]-\Tr[\sigma\log\rho]$.
By definition, $D(\sigma(0)\|\rho)=0$.
The logarithm of a positive operator $\sigma$ can be written as
\begin{equation}
    \log \sigma(t) = \int_0^\infty \qty( \frac{1}{1+u} - (\sigma(t)+u)^{-1} ) du.
\end{equation}
Because the derivative of the inverse of some invertible matrix $A$ is given by
$\partial_t A^{-1}=-A^{-1} (\partial_t A) A^{-1}$, we have
\begin{equation}
    \partial_t(\sigma(t)+u)^{-1} = - (\sigma(t)+u)^{-1} \,\dot{\sigma}(t) \, (\sigma(t)+u)^{-1},\qquad \dot{\sigma}=\partial_t\,\sigma,
\end{equation}
which yields
\begin{equation}
    \partial_t\log\sigma=\int_0^\infty\! du\,(\sigma+u)^{-1}\dot\sigma\,(\sigma+u)^{-1}.
    \label{eq:log-derivative}
\end{equation}

Because the trace cyclicity implies that
\begin{equation}
    \Tr[\sigma \,\partial_t \log \sigma] = \Tr \dot{\sigma}=\partial_t \Tr\sigma =\partial_t (1)=0,
\end{equation}
hence
\begin{equation}
    \partial_t D(\sigma(t)\|\rho)=\Tr[\dot\sigma\,(\log\sigma-\log\rho)] \quad \therefore \ \partial_t D(\sigma(t)\|\rho)\vert_{t=0}=0.
\end{equation}
Differentiating once more and setting $t=0$, the term proportional to $\ddot\sigma$ drops because $\log\sigma-\log\rho$ vanishes there, leaving
\begin{equation}
    D\big(\rho(g+dg)\big\|\rho(g)\big)=\tfrac12\chi_{ij}\,dg^idg^j+O(dg^3),
    \qquad
    \chi_{ij}=\Tr\!\left[\partial_i\rho\,\partial_j\log\rho\right],
    \label{eq:BKM-definition}
\end{equation}
which is Eq.~\eqref{eq:BKM-expansion}. Substituting Eq.~\eqref{eq:log-derivative}, we obtain the BKM (Kubo--Mori) metric~\cite{petz2008quantum},
\begin{equation}
    \chi_{ij}=\int_0^\infty\! du\,\Tr\!\left[\partial_i\rho\,(\rho+u)^{-1}\partial_j\rho\,(\rho+u)^{-1}\right],
    \label{eq:BKM-resolvent}
\end{equation}
which is manifestly symmetric and positive semidefinite.

\subsection*{2. Formula for \texorpdfstring{$\partial_i\rho$}{partial i rho}}

The derivative of a thermal state $\rho=e^{-\beta H}/Z$ is calculated as follows. It follows that
\begin{equation}
    \frac{\partial_i e^{-\beta H}}{Z}
    =-\frac{1}{Z}\int_0^\beta\! d\tau\;e^{-\tau H}\,(\partial_iH)\,e^{-(\beta-\tau)H}
    =-\frac{\beta}{Z}\int_0^1\! ds\;e^{-\beta sH}\,h_i\,e^{-\beta(1-s)H}
    =-\beta\int_0^1\! ds\;\rho^s h_i\rho^{1-s}=-\beta\,\Omega_\rho(h_i)
\end{equation}
where $h_i:=\partial_iH$ and $\Omega_\rho(h_i)=\int_0^1ds\,\rho^s h_i\rho^{1-s}$.

By taking the trace, we obtain
\begin{equation}
    \frac{\partial_i Z}{Z} = -\beta \Tr \Omega_\rho(h_i) = -\beta \expval{h_i}_\rho \Rightarrow \partial_i Z^{-1}=\frac{\beta \expval{h_i}_\rho}{Z}
    \quad\therefore
    e^{-\beta H} \partial_i Z^{-1} = \beta \expval{h_i}_\rho \rho = \beta \Omega_\rho(\expval{h_i}_\rho)
\end{equation}
where we used $\Omega_\rho(\id)=\rho$ in the last equality.
Therefore
\begin{equation}
    \partial_i\rho=-\beta\,\Omega_\rho(\delta h_i),
    \qquad
    \delta h_i:=h_i-\expval{h_i}_\rho,
    \label{eq:duhamel-drho}
\end{equation}
as quoted in the main text.

\subsection*{3. Calculation of the BKM metric \texorpdfstring{$\chi_{ij}$}{chi ij}}

Because $\log\rho=-\beta H-\log Z$ and $\partial_j\log Z=-\beta\expval{h_j}_\rho$,
\begin{equation}
    \partial_j\log\rho=-\beta h_j+\beta\expval{h_j}_\rho=-\beta\,\delta h_j .
    \label{eq:dlogrho}
\end{equation}
Substituting Eqs.~\eqref{eq:duhamel-drho} and~\eqref{eq:dlogrho} into Eq.~\eqref{eq:BKM-definition},
\begin{equation}
    \chi_{ij}
    =\beta^2\,\Tr\!\left[\Omega_\rho(\delta h_i)\,\delta h_j\right]
    =\beta^2\int_0^1\! ds\;\Tr\!\left[\rho^{s}\,\delta h_i\,\rho^{1-s}\,\delta h_j\right].
    \label{eq:chi-duhamel}
\end{equation}

Let $A(\tau):=e^{\tau H}Ae^{-\tau H}$ denote Euclidean Heisenberg evolution. Then
\begin{equation}
    \expval{\delta h_i(\tau)\,\delta h_j(0)}_\beta
    =\frac1Z\Tr\!\left[e^{-\beta H}e^{\tau H}\delta h_ie^{-\tau H}\delta h_j\right]
    =\frac1Z\Tr\!\left[e^{-(\beta-\tau)H}\,\delta h_i\,e^{-\tau H}\,\delta h_j\right].
    \label{eq:delta-h-H}
\end{equation}
With $\tau=\beta(1-s)$, this equals $\Tr[\rho^s\delta h_i\rho^{1-s}\delta h_j]$, the integrand of Eq.~\eqref{eq:chi-duhamel}. 
Hence, Eq.~\eqref{eq:chi-duhamel} becomes
\begin{equation}
    \chi_{ij}=\beta\int_0^\beta\! d\tau\;\expval{\delta h_i(\tau)\,\delta h_j(0)}_\beta .
    \label{eq:chi-euclidean}
\end{equation}
The symmetry $\chi_{ij}=\chi_{ji}$ is realized here by the KMS condition $\expval{\delta h_i(\tau)\delta h_j(0)}_\beta=\expval{\delta h_j(0)\delta h_i(\tau-\beta)}_\beta$.

Let us now substitute the explicit coupling dependence. The Hamiltonian is given as $H(g)=H_0+\sum_i g^i\int_\Sigma d^{d-1}\bm x\,\cO_i(0,\bm x)$. Thus,
\begin{equation}
    h_i=\pdv{H}{g^i}=\int_\Sigma\! d^{d-1}\bm x\;\cO_i(0,\bm x),
    \qquad
    \delta h_i(\tau)=\int_\Sigma\! d^{d-1}\bm x\;\left[\cO_i(\tau,\bm x)-\expval{\cO_i(\tau,\bm x)}_\rho\right].
\end{equation}
Substituting this in Eq.~\eqref{eq:chi-euclidean}, we obtain Eq.~\eqref{eq:BKM-Euclidean-correlator},
\begin{equation}
    \chi_{ij}=\beta\int_0^\beta\! d\tau\!\int_\Sigma\! d^{d-1}\bm x\,d^{d-1}\bm y\;\expval{\cO_i(\tau,\bm x)\,\cO_j(0,\bm y)}_{\beta,c},
\end{equation}
where the connected two-point function is defined as
\begin{equation}
    \expval{AB}_{\beta,c}=\expval{AB}_\rho - \expval{A}_\rho\expval{B}_\rho.
\end{equation}

\subsection*{4. Zero-temperature limit of the susceptibility}
Let us now evaluate \eqref{eq:chi-euclidean} in the zero temperature limit,up to the divergent prefactor $\beta$.
For a finite regulated system with a nondegenerate ground state $|0\rangle$, let $\epsilon_n=E_n-E_0>0$, $p_n=e^{-\beta E_n}/Z$, and $h_v=v^i h_i$. Based on \eqref{eq:delta-h-H}, the ground--excited contributions to the connected thermal correlator are
\begin{equation}
    \expval{\delta h_v(\tau)\delta h_v(0)}_\beta
    =\sum_{n>0}p_0|\langle n| h_v|0\rangle|^2
    \left[e^{-\tau\epsilon_n}+e^{-(\beta-\tau)\epsilon_n}\right]+R_\beta(\tau).
    \label{eq:thermal-endpoint-pair}
\end{equation}
The remainder includes the diagonal connected fluctuations and the terms involving only excited states; $\int_0^\beta d\tau\,R_\beta(\tau)\to0$ at fixed finite dimension. Since each exponential in Eq.~\eqref{eq:thermal-endpoint-pair} integrates to $(1-e^{-\beta\epsilon_n})/\epsilon_n$, Eq.~\eqref{eq:chi-euclidean} yields
\begin{equation}
    \lim_{\beta\to\infty}\frac{\chi_{vv*}}{\beta}
    =2\sum_{n>0}\frac{|\langle n| h_v|0\rangle|^2}{\epsilon_n}
    =2\int_0^\infty d\tau\,
    \expval{\delta h_v(\tau)\delta h_v(0)}_{\beta=\infty}.
    \label{eq:BKM-zero-temperature-spectral}
\end{equation}
For a Hamiltonian perturbation $H_*+u h_v$, this equals $-\partial_u^2E_0|_{u=0}$, the static ground-state susceptibility.

Nondegeneracy is a sufficient condition for this finite-system argument. In a degenerate ground space with projector $P_0$, a nonzero variance within the ground space remains and leads to a $\beta^2$ contribution from $\beta\int_0^\beta d\tau$ to $\chi_{vv*}$. This leads to a divergence even if we divide by $\beta$. We require the absence of this residual ground-space fluctuation.

At a fixed spatial circumference $L$, the zero-temperature geometry is the cylinder $\mathbb R\times S^1_L$. 
To derive the susceptibility density on this cylinder, define the connected, Euclidean time-ordered thermal correlator
\begin{equation}
    C_{\beta,L}(\tau,x)
    :=\expval{\mathcal T_\tau\,\cO_v(\tau,x)\cO_v(0,0)}_{\beta,L,c}.
\end{equation}
For $0<\tau<\beta$, time ordering leaves the operator product unchanged. Since $h_v=\int_0^L dx\,\cO_v(0,x)$, Eq.~\eqref{eq:chi-euclidean} involves an integral over both spatial insertion points. Spatial translation invariance and periodic boundary conditions (which naturally follow in conformal field theory) give
\begin{equation}
    \int_0^L dx\int_0^L dy\,C_{\beta,L}(\tau,x-y)
    =L\int_0^L dx\,C_{\beta,L}(\tau,x).
\end{equation}
Thus one spatial integral supplies the factor $L$ removed in the density $\chi_{vv*}/(\beta L)$. The KMS condition for the bosonic (scalar) operator $\cO_v$ identifies $C_{\beta,L}(\tau,x)=C_{\beta,L}(\tau+\beta,x)$ for $-\beta<\tau<0$. Splitting the thermal time integral at $\beta/2$ and shifting its second half by $-\beta$ therefore yields
\begin{equation}
\begin{aligned}
    \frac{\chi_{vv*}}{\beta L}
    &=\int_0^\beta d\tau\int_0^L dx\,C_{\beta,L}(\tau,x)\\
    &=\int_{-\beta/2}^{\beta/2}d\tau\int_0^L dx\,C_{\beta,L}(\tau,x).
\end{aligned}
\end{equation}
As $\beta\to\infty$, $C_{\beta,L}(\tau,x)$ tends to the vacuum correlator $C_{0,L}(\tau,x)$. Assuming that this limit commutes with the spacetime integral, we obtain
\begin{equation}
    \lim_{\beta\to\infty}\frac{\chi_{vv*}}{\beta L}
    =\int_{-\infty}^{\infty}d\tau\int_0^L dx\,
    \expval{\mathcal T_\tau\,\cO_v(\tau,x)\cO_v(0,0)}_{0,L,c}.
    \label{eq:BKM-vacuum-cylinder}
\end{equation}
Taking $L\to\infty$ after this limit, at fixed UV cutoff $\epsilon$ and reference scale $\mu$, replaces the cylinder correlator by the plane correlator in Eq.~\eqref{eq:BKM-Zam}.

\subsection*{5. Circular point splitting and fixed-point metric matching}

Let the physical RG tangent $v$ correspond to a scalar scaling operator $\Phi_v$ at the two-dimensional IR conformal fixed point. For $S=S_*+\lambda\int d^2z\,\Phi_v(z)$, the mass dimension of $\lambda$ is $2-\Delta_v$. The dimensionless coupling $u=\lambda\mu^{\Delta_v-2}$ obeys $du/ds=-(\Delta_v-2)u+O(u^2)$ for $s=-\log\Lambda$. Thus the stable exponent in the main text is
\begin{equation}
    \omega=\Delta_v-2>0,
    \qquad \cO_v=\mu^{2-\Delta_v}\Phi_v=\mu^{-\omega}\Phi_v.
    \label{eq:stable-scaling-dimension}
\end{equation}
Here $\mu$ is the fixed reference scale at which the coupling coordinates and metric are compared; $\Lambda$ parametrizes the physical RG flow. The scaling dimension includes anomalous dimensions, and operator mixing is diagonalized to select this scaling direction. The normalization in Eq.~\eqref{eq:Zam-fixed-point-normalization} then follows from the fixed-point two-point function. In the convention of Ref.~\cite{Friedan:2009ik}, an additional factor $6\pi^2$ multiplies the local metric; absorbing it into $\kappa$ gives the convention used here.

Adopting the UV regularization $|z|>\epsilon$ in the complex coordinate $z=x+i\tau$, the BKM metric per unit spacetime volume in the zero-temperature and thermodynamic limits is given by
\begin{equation}
    \hat{\chi}_{vv*}
    :=\lim_{L\to\infty}\lim_{\beta\to\infty}\frac{\chi_{vv*}}{\beta L}
    =\int_{|z|>\epsilon}d^2z\,
    \expval{\cO_v(z)\cO_v(0)}_{0,c}.
    \label{eq:circular-BKM-matching}
\end{equation}
Integrating over an annulus $\epsilon<|z|<R$ using Eq.~\eqref{eq:Zam-fixed-point-normalization} with $\Delta_v=2+\omega$ gives
\begin{equation}
\begin{aligned}
    \int_{\epsilon<|z|<R}d^2z\,
    \expval{\cO_v(z)\cO_v(0)}_{0,c}
    &=2\pi\mu^{-2\omega}g^{\rm Zam}_{vv*}
        \int_\epsilon^R dr\,r^{1-2\Delta_v}\\
    &=\frac{\pi\mu^2}{\Delta_v-1}
        \left[(\mu\epsilon)^{-2(\Delta_v-1)}-(\mu R)^{-2(\Delta_v-1)}\right]
        g^{\rm Zam}_{vv*}.
\end{aligned}
    \label{eq:circular-annulus-integral}
\end{equation}
The $R\to\infty$ limit converges since $\Delta_v=2+\omega>2$. Taking $R\to\infty$ proves Eq.~\eqref{eq:BKM-Zam} of the main text:
\begin{equation}
    \hat{\chi}_{vv*}
    =\frac{\pi\mu^2}{(\Delta_v-1)(\mu\epsilon)^{2(\Delta_v-1)}}g^{\rm Zam}_{vv*}
    =\mu^2 C_v(\mu\epsilon)\,g^{\rm Zam}_{vv*},
    \qquad
    C_v(\mu\epsilon):=\frac{\pi}{(\Delta_v-1)(\mu\epsilon)^{2(\Delta_v-1)}}.
    \label{eq:dimensionless-BKM-Zam}
\end{equation}

\subsection*{6. Resource density and the $c$-function}

At each finite $\beta$, $L$, and regulator, keep $V$ and $\widehat\cG$ fixed along the thermal curve and require $\widehat\cG_\infty(\rho(g(u)))=\rho_*>0$. The embedded state and its Ces\`aro mean are faithful on $P\cH$. Equation~\eqref{eq:rg-supported-resource} therefore identifies the resource with $D(\rho(g(u))\Vert{}\rho_*)$, and the BKM expansion proves Eq.~\eqref{eq:I-Zam}.

Along $g^i(u)=g_*^i+uv^i+O(u^2)$ with $du/ds=-\omega u+O(u^2)$, Eq.~\eqref{eq:c-function-slope} yields
\begin{equation}
    \frac{dc}{du}
    =-\kappa g^{\rm Zam}_{ij}(g(u))\frac{dg^i}{du}\frac{dg^j}{du}\frac{du}{ds}
    =\kappa\omega g^{\rm Zam}_{vv*}u+O(u^2).
\end{equation}
Integration from $0$ to $u$ proves Eq.~\eqref{eq:c-near-IR-direction}, $c(u)-c_{\rm IR}=\frac{\kappa\omega}{2}g^{\rm Zam}_{vv*}u^2+O(u^3)$.

Dividing Eq.~\eqref{eq:I-Zam} by $\beta L$ and substituting Eq.~\eqref{eq:dimensionless-BKM-Zam} in the $L,\beta\to\infty$ limit gives Eq.~\eqref{eq:IG-Zam-density}:
\begin{equation}
    \lim_{L,\beta\to\infty}\frac{\cI_\cG(V\rho(g(u))V^\dagger)}{\beta L}
    =\frac{1}{2}\hat{\chi}_{vv*}u^2+O(u^3)
    =\frac{1}{2}\mu^2 C_v(\mu\epsilon)\,g^{\rm Zam}_{vv*}u^2+O(u^3).
\end{equation}
Combining this with Eq.~\eqref{eq:c-near-IR-direction} eliminates $\frac{1}{2}g^{\rm Zam}_{vv*}u^2$ and yields Eq.~\eqref{eq:IG-c-near-IR}:
\begin{equation}
    \lim_{L,\beta\to\infty}\frac{\cI_\cG(V\rho(g(u))V^\dagger)}{\beta L}
    =\frac{\mu^2 C_v(\mu\epsilon)}{\kappa\omega}\bigl(c(u)-c_{\rm IR}\bigr)+O(u^3).
\end{equation}

\end{document}